\documentclass{amsart}
\usepackage{fullpage}
\usepackage{times}
\usepackage{color}
\usepackage{amsmath,amssymb,amsthm, amsfonts, enumitem,centernot,mathtools,extarrows,mathdots,appendix}
\usepackage{comment}
\usepackage[ruled,vlined]{algorithm2e}
\usepackage{tikz-cd}

\newtheorem{thm}{Theorem}[section]
\newtheorem{theorem}[thm]{Theorem}
\newtheorem{lemma}[thm]{Lemma}

\newtheorem{proposition}[thm]{Proposition}
\newtheorem{corollary}[thm]{Corollary}
\newtheorem{eg}{Example}[section]

\newtheorem{fact}{Fact}

\newcommand{\ket}[1]{|#1\rangle}
\newcommand{\bra}[1]{\langle #1|}
\newcommand{\braket}[2]{\left \langle #1 \middle| #2 \right \rangle}

\numberwithin{equation}{section} \errorcontextlines=0

\begin{document}
\title{Qubits as Hypermatrices and Entanglement}
\author{Isaac Dobes and Naihuan Jing}
\address{Department of Mathematics, North Carolina State University, Raleigh, NC 27695, USA}
\email{idobes@ncsu.edu}
\email{jing@math.ncsu.edu}
\date{\vspace{-5ex}}

\begin{abstract}
    In this paper, we represent $n$-qubits as hypermatrices and consider various applications to quantum entanglement. In particular, we use the higher-order singular value decomposition of hypermatrices to prove that the $\pi$-transpose is an LU invariant. Additionally, through our construction we show that the matrix representation of the combinatorial hyperdeterminant of $2n$-qubits can be expressed as a product of the second Pauli matrix, allowing us to derive a formula for the combinatorial hyperdeterminant of $2n$-qubits in terms of the $n$-tangle. 
\end{abstract}
\maketitle

\section{Introduction}
For the last few decades, classifying entangled states has been a major endeavor for researchers in theoretical quantum information \cite{CN,BZ,YM,MW}. For bipartite quantum systems, the theory of entanglement is well understood and established \cite{Concurrence}, however for multipartite systems, the very notion of entanglement is still being worked out \cite{EGW,ZVSW, SLOCC,CJ}. As such, much of the focus has been on better understanding and expanding the theory of entanglement in multipartite systems \cite{EGW}. 

Two pure states are considered equivalently entangled if they are \textit{locally unitarily (LU) equivalent}; if $\ket{\psi}$ and $\ket{\varphi}$ are $n$-qubit states, then this means that there exist $U_1,...,U_n\in SU(2)$ such that 
\begin{equation}
\ket{\varphi} = (U_1\otimes...\otimes U_n)\ket{\psi}.
\end{equation}
Thus, it is of great importance to find operations on states that are invariant under local unitary equivalence in the classification of entangled states. In this paper, we represent pure $n$-qubits as hypermatrices and apply the theory of multilinear algebra to these states to study LU invariants. Specifically, we consider the higher-order singular value decomposition of hypermatrices \cite{LMV,CZCH} and show from our representation that the $\pi$-transpose is an LU-invariant. Next, we prove a formula relating the matrix of the hyperdeterminant of an arbitrary $2n$-qubit to the tensor product of the second Pauli matrix, which then allows us to express the $n$-tangle \cite{CW} in terms of the hyperdeterminant. This shows that in some sense, the hyperdeterminant provides a measurement of entanglement.


\section{Preliminaries}\label{S2}
Let $\mathbb C^n$ be the complex $n$-dimensional vector space. Let $v_i\in \mathbb{C}^{n_i}$ be $N$ vectors, where $(v_i)_j$ are the
$j$th coordinates of $v_i$. The \textit{outer product} of $v_1$,$v_2,...,v_N$ is defined to be the hypermatrix $v_1\circ v_2\circ...\circ v_N\in \mathbb{C}^{n_1\times n_2\times...\times n_N}$ of order $N$ whose $(i_1i_2...i_N)$-coordinate is given by $(v_1)_{i_1}(v_2)_{i_2}...(v_N)_{i_N}$. 

Now, let $H\in \mathbb{C}^{n_1\times n_2\times...\times n_N}$ be a hypermatrix 
and $A_i\in \mathbb{C}^{m_i\times n_i}$ be $N$ rectangular matrices. The \textit{multilinear multiplication} of $(A_1,A_2,...,A_N)$ with $H$ is defined to be the hypermatrix
$(A_1,A_2,...,A_N)*H =: H'$, 
where 
\begin{equation}H'_{i_1i_2...i_N} = \sum\limits_{j_1,j_2,...,j_N=1}^{n_1,n_2,...,n_N}(A_1)_{i_1j_1}(A_2)_{i_2j_2}...(A_N)_{i_Nj_N}H_{j_1j_2...j_N}.
\end{equation}
Multilinear multiplication is linear in terms of the matrices in both parts; that is, if $\alpha,\beta\in \mathbb{C}$, $A_1,B_1\in \mathbb{C}^{m_1\times n_1}$; $A_2,B_2\in\mathbb{C}^{m_2\times n_2}$;...; $A_N,B_N\in \mathbb{C}^{m_N\times n_N}$; and $H,K\in \mathbb{C}^{n_1\times n_2\times...\times n_N}$; then 
\begin{equation}
(A_1,A_2,...,A_N)*(\alpha H + \beta K) = \alpha(A_1,A_2,...,A_N)*H + \beta(B_1,B_2,...,B_N)*K
\end{equation}
and 
\begin{equation}
[\alpha(A_1,A_2,...,A_N)+\beta(B_1,B_2...,B_N)]*H = \alpha(A_1,A_2,...,A_N)*H + \beta(B_1,B_2,...,B_N)*H.
\end{equation}
Multilinear multiplication interacts with the outer product in the following way:
\begin{equation}
(A_1,A_2,...,A_N)*\left(\sum\limits_{k=1}^r\alpha_k(v_1^{(k)}\circ v_2^{(k)}\circ...\circ v_N^{(k)})\right) = \sum\limits_{k=1}^r\alpha_k(A_1v_1^{(k)})\circ (A_2v_2^{(k)})\circ...\circ (A_Nv_N^{(k)})
\end{equation}
(where $\alpha_k\in \mathbb{C}$ and $v_i^{(k)}\in \mathbb{C}^{n_i}$). For more details on these and other operations on hypermatrices, the reader is referred to \cite{LH}.

We also need the notion of higher-order singular value decomposition. 
The \textit{higher-order singular value decomposition}, first discovered in \cite{LMV}, states that any hypermatrix $H\in \mathbb C^{n_1\times n_2\cdots\times n_N}$ can be written as
\begin{equation}H = (V_1,V_2,...,V_N)*\Sigma
\end{equation}
where each $V_k$ is an $n_k\times n_k$ unitary matrix for $1\leq k\leq N$ and $\Sigma$ is an $n_1\times n_2\times...\times n_N$ hypermatrix 
such that for each $\Sigma_{i_k=\alpha}$, obtained by fixing the $k^{th}$ index to $\alpha$, satisfies:
\begin{enumerate}
    \item the all-orthogonality that $\langle \Sigma_{i_k=\alpha},\Sigma_{i_k=\beta}\rangle = 0$ for all $1\leq k\leq N$ and $\alpha\neq \beta$, where $\langle\text{ , }\rangle$ is the Frobenius inner product, and 
    \item the ordering that $\|\Sigma_{i_k=1}\|\geq\|\Sigma_{i_k=2}\|\geq...\geq\|\Sigma_{i_k=n_k}\|\geq 0$ for $1\leq k\leq N$, where $\|\cdot\|$ is the Frobenius norm.
\end{enumerate}
We call $\Sigma$ a {\it core tensor} of $H$, and $\Sigma_{i_k=j}$ {\it subtensors} of $\Sigma$. We also call $\|\Sigma_{i_k=j}\|:=\sigma_j^{(k)}$ the {\it $k$-mode singular values} of $H$. It is known that
the $k$-mode singular values are unique \cite{LMV}; that is,
\begin{equation}
H\mapsto \{\text{$k$-mode singular values of H}\}
\end{equation}
is a well-defined function. Note that When $N=2$, the higher-order singular value decomposition reduces to the typical matrix singular value decomposition. 

Indeed, we may express the higher-order singular value entirely in terms of matrices by considering the $k$-mode unfolding. Recall that the \textit{$k$-mode unfolding} \cite{KB} of a hypermatrix $H\in \mathbb{C}^{n_1\times n_2\times...\times n_N}$ is the $n_k\times (n_{k+1}...n_Nn_1...n_{k-1})$ matrix, denoted $H_{(k)}$, whose $(i_k,j)$ entry is given by $(i_1,...,i_N)$-entry of $H$, with
\begin{equation}
j=1+\sum\limits_{\substack{l=1 \\ l\neq k}}^N\left[(i_l-1)\prod\limits_{\substack{m=1 \\ m\neq k}}^{l-1}n_m\right],
\end{equation}
or in the case where the index starts at $0$, 
\begin{equation}
    j = \sum\limits_{\substack{l=1 \\ l\neq k}}^N\left[i_l\prod\limits_{\substack{m=1 \\ m\neq k}}^{l-1}n_m\right].
\end{equation}
For instance, if $H = [H_{i_1i_2i_3}]\in \mathbb{C}^{2\times 2\times 2}$, then $H_{(1)}$ is the $2\times 4$ matrix given by
\begin{equation}
H_{(1)} = \left[\begin{array}{cccc}
    H_{111} & H_{121} & H_{112} & H_{122} \\
    H_{211} & H_{221} & H_{212} & H_{222}
\end{array}\right].
\end{equation}
It was shown in \cite{LMV} that if $H\in \mathbb{C}^{n_1\times n_2\times...\times n_N}$ has the higher-order singular value decomposition 
\begin{equation}
    H = (V_1,...,V_n)*\Sigma,
\end{equation}
then $H_{(n)}$ has the matrix singular value decomposition
\begin{equation}
    H_{(n)} = V_n\Sigma_{(k)}(V_{k+1}\otimes...\otimes V_N\otimes V_1\otimes...\otimes V_{k-1}),
\end{equation}
where $\Sigma_{(k)} = \mathrm{diag}(\sigma_1^{(k)},...,\sigma_{n_k}^{(k)})\in \mathbb{C}^{n_k\times (n_{k+1}...n_Nn_1...n_{k-1})}$. 

We also review the notion of the {\it $\pi$-transpose} \cite{LH} of a hypermatrix $H = [H_{i_1...i_N}]\in \mathbb{C}^{n_1\times n_2\times...\times n_N}$, which is defined as the hypermatrix
\begin{equation}
H^{\pi} := [H_{\pi(i_1)...\pi(i_N)}]\in \mathbb{C}^{n_{\pi(1)}\times n_{\pi(2)}\times...\times n_{\pi(N)}},
\end{equation}
where $\pi\in S_N$. We also note that if $n_1=n_2=...=n_N =: n$, then we say that $H$ is a \textit{cuboid hypermatrix of order N with length n}.  

Lastly, we review the Cayley's first hyperdeterminant, also known as the combinatorial hyperdeterminant. Suppose $H$ is a cuboid hypermatrix of order $N$ with side length $n$, i.e. $H\in \mathbb{C}^{\overbrace{n\times...\times n}^{N\text{ times}}}$. For a permutation $\sigma\in S_n$, let $l(\sigma)=l$ denote the smallest number of transpositions needed to form $\sigma$: $\sigma=s_{i_1}\ldots s_{i_l}$. Then the \textit{combinatorial hyperdeterminant} \cite{LH}, of $H$ is defined to be
\begin{equation}\mathrm{hdet}(H) := \frac{1}{n!}\sum\limits_{\sigma_1,\sigma_2,...,\sigma_N\in S_n}(-1)^{\sum\limits_{i=1}^nl(\sigma_j)}\prod\limits_{j=1}^NA_{\sigma_1(j)\sigma_2(j)...\sigma_N(j)}.
\end{equation}
Note that $\text{hdet}$ is identically $0$ for all hypermatrices of odd order, and for hypermatrices of even order it is equal to
\begin{equation}\sum\limits_{\sigma_2,...,\sigma_N\in S_n}(-1)^{\sum\limits_{i=2}^nl(\sigma_j)}\prod\limits_{j=1}^NA_{j\sigma_2(j)...\sigma_N(j)}
\end{equation}
(\cite{LH,AY}). The next result is well-known and referenced in this paper.
\begin{proposition}\label{p:prop1}\cite{LH,AY}
	For $A_1,A_2,...,A_N\in SL(n)$, 
 \begin{equation}
 \mathrm{hdet}((A_1,A_2,...,A_N)*H) = \mathrm{hdet}(H).
 \end{equation}
\end{proposition}

\section{Correspondance Between Qubits and Hypermatrices}\label{S3}
Suppose we have two strings $a  = a_1a_2...a_n$ and $b = b_1b_2...b_n$. Recall that $a<b$ in the lexicographic order if $a_i<b_i$, where $i$ is the first position where the two strings differ. For example, in the lexicographic order, $000 < 001 < 010 < 011 < 100 < 101 < 110 < 111$.

Let $\psi$ be any pure $n$-qubit state 
${\psi}= \sum\limits_{i_1,...,i_n\in \{0,1\}}\psi_{i_1...i_n}\ket{i_1...i_n}\in (\mathbb{C}^2)^{\otimes n} \cong \mathbb{C}^{2^n}$, where
$\ket{i_1...i_n} = \ket{i_1}\otimes...\otimes\ket{i_n}$, and the amplitudes satisfy 
\begin{equation}\sum\limits_{i_1,...,i_n\in \{0,1\}}|\psi_{i_1...i_n}|^2=1.
\end{equation}
We can order the amplitudes of $\psi$ in the lexicographic order and define the $2^n$-dimensional vector
\begin{equation}
\ket{\psi} = (\psi_{0...00}, \psi_{0...01},\psi_{0...10},\psi_{0...11}, \ldots, \psi_{1...11})^t.\end{equation}
In the following, we will identify the pure state $\psi$ with the vector $\ket{\psi}$. We also consider
the following outer product 
\begin{equation}\widehat{\psi}=\sum\limits_{i_1,...,i_n\in \{0,1\}}\psi_{i_1...i_n}\ket{i_1}\circ...\circ\ket{i_n}
\end{equation}
a cuboid hypermatrix of length $2$ and order $n$, whose Frobenius norm is $1$. In other words, 
\begin{equation}\widehat{\psi} = [\psi_{i_1...i_n}]_{2\times ...\times 2}
\end{equation}
with the $((i_1+1),...,(i_n+1))$-entry of $\widehat{\psi}$ entry being $\psi_{i_1...i_n}$. Consequently, we have an isomorphism between the Hilbert spaces of pure $n$-qubits $\psi$ and their corresponding hypermatrices $\widehat{\psi}$.

Let $\psi$ and $\varphi$ be two pure $n$-qubit states with corresponding hypermatrices $\widehat{\psi}$ and $\widehat{\varphi}$. 
We say that two hypermatrices $\widehat{\psi}$ and $\widehat{\varphi}$ are {\it LU equivalent} if there exists $U_1,...,U_n\in SU(2)$ such that
\begin{equation}
\widehat{\varphi} = (U_1,...,U_n)*\widehat{\psi}.
\end{equation}

\begin{lemma}\label{l:LU} 
    The pure states
    $\psi$ and $\varphi$ are LU equivalent if and only if $\widehat{\psi}$ and $\widehat{\varphi}$ are LU equivalent. 
\end{lemma}
\begin{proof} 
    Suppose $\psi$ and $\varphi$ are LU equivalent, then there exists $U_1,...,U_n\in SU(2)$ such that 
    \[\ket{\varphi} = (U_1\otimes...\otimes U_n) \ket{\psi}.\]
    Observe 
    \begin{align*}
        (U_1\otimes...\otimes U_n) \ket{\psi} &= (U_1\otimes...\otimes U_n)\sum\limits_{i_1,...,i_n\in \{0,1\}}\psi_{i_1...i_n}\ket{i_1}\otimes...\otimes \ket{i_n} \\
        &= \sum\limits_{i_1,...,i_n\in \{0,1\}}\psi_{i_1...i_n}(U_1\ket{i_1})\otimes...\otimes (U_n\ket{i_n})
    \end{align*}
    where the last equality follows from the linearity of Kronecker products. The isomorphism constructed above maps this vector to the hypermatrix
    \begin{equation}
    \widehat{\varphi} = \sum\limits_{i_1,...,i_n\in \{0,1\}}\psi_{i_1...i_n}(U_1\ket{i_1})\circ...\circ (U_n\ket{i_n})
    \end{equation}
    and by the linearity of multilinear matrix multiplication, this is equal to
    \begin{equation}(U_1,...,U_n)*\left(\sum\limits_{i_1,...,i_n\in \{0,1\}}\psi_{i_1...i_n}\ket{i_1}\circ...\circ \ket{i_n}\right)
    \end{equation}
    which is precisely
    \begin{equation}
    \widehat{\varphi}=(U_1,...,U_n)*\widehat{\psi}.
    \end{equation}
Thus, $\psi$ and $\varphi$ are LU equivalent if and only if $\widehat{\psi}$ and $\widehat{\varphi}$ are LU equivalent. 
\end{proof}
We note that our construction and this result straightforwardly extends to $n$-qudits, however, for this paper we only focus on $n$-qubits.

Consequently, if $\psi$ and $\varphi$ are LU equivalent, then 
\[\widehat{\varphi} = (U_1,...,U_n)*\widehat{\psi}\]
for some $U_1,...,U_n\in SU(2)$, and so if $\widehat{\psi}$ has the higher-order singular value decomposition
\begin{equation}
    \widehat{\psi} = (V_1,...,V_n)*\Sigma,
\end{equation}
for some $V_1,...,V_n\in U(2)$, then 
\begin{equation}
    \widehat{\varphi} = (U_1,...,U_n)*\big((V_1,...,V_n)*\Sigma\big) = (U_1V_1,...,U_nV_n)*\Sigma
\end{equation}
is the higher-order singular value decomposition for $\widehat{\varphi}$, showing that they share the same core tensor. Hence, they have the same $k$-mode singular values (by uniqueness). On the other hand, in \cite{Qiao1, Qiao2}, Liu et. al. proved that if two states $\psi$ and $\varphi$ have the same core tensor, then they are LU equivalent. We thus have the following theorem.
\begin{theorem}
    For any $\pi\in S_n$, $\widehat{\psi}$ and $\widehat{\psi}^{\pi}$ are LU equivalent. 
\end{theorem}
\begin{proof}
    Suppose $\widehat{\psi}$ has the higher-order singular value decomposition
    \[(V_1,...,V_n)*\Sigma\]
    for some $V_1,...,V_n\in U(2)$. In particular, this implies that the $((i_1+1),...,(i_n+1))$-entry of $\widehat{\psi}$ is given by the sum
    \[\sum\limits_{j_1,...,j_n=1}^2(V_1)_{(i_1+1)j_1}...(V_n)_{(i_n+1)j_n}\Sigma_{j_1...j_n}\]
    and so for any $\pi\in S_n$, the $((i_1+1),...,(i_n+1))$-entry of $\psi^{\pi}$ is given by
    \begin{equation}
        \psi_{\pi(i_1)...\pi(i_n)} = \sum\limits_{j_1,...,j_n=1}^2(V_1)_{\pi(i_1+1)j_1}...(V_n)_{\pi(i_n+1)j_n}\Sigma_{j_1...j_n}.
    \end{equation}  
    Since $i_k\in \{0,1\}$ for each $k\in [n]$, the above sum is well-defined; moreover, since we are just permuting the rows in the sum, it follows that 
    \begin{equation}
        \widehat{\psi}^{\pi} = (P_1V_1,...,P_nV_n)*\Sigma,
    \end{equation}
    where $P_1,..., P_n$ are some $2\times 2$ permutation matrices (which recall are orthogonal, hence unitary). Thus, $\widehat{\psi}$ and $\widehat{\psi}^{\pi}$ have the same core tensor in their higher-order singular value decomposition, proving that they are LU equivalent.
\end{proof}
\begin{eg}
    Consider the $3$-qubit $\ket{\psi} = \frac{1}{2}\ket{000}-\frac{1}{2}\ket{100}+\frac{1}{\sqrt{2}}\ket{101} = \left[\begin{array}{cccccccc}
        \frac{1}{2} & 0 & 0 & 0 & -\frac{1}{2} & \frac{1}{\sqrt{2}} & 0 & 0 
    \end{array}\right]^t$. The corresponding $2\times 2\times 2$ hypermatrix $\widehat{\psi}$ is given by
    \[\widehat{\psi} = \left[\begin{array}{cc|cc}
        \frac{1}{2} & 0 & 0 & 0 \\
        -\frac{1}{2} & 0 & \frac{1}{\sqrt{2}} & 0
    \end{array}\right]\]
    which in matrix form can be represented as
    \[\widehat{\psi}_{(1)} = \left[\begin{array}{cccc}
        \frac{1}{2} & 0 & 0 & 0 \\
        -\frac{1}{2} & 0 & \frac{1}{\sqrt{2}} & 0
    \end{array}\right].\]
    For $\pi_1=(13)\in S_3$, the corresponding $2\times 2\times 2$ hypermatrix is given by
    \[\widehat{\psi}^{\pi_1} = \left[\begin{array}{cc|cc}
        \frac{1}{2} & 0 & -\frac{1}{2} & 0 \\
        0 & 0 & \frac{1}{\sqrt{2}} & 0
    \end{array}\right]\]
    which in matrix form can be represented as
    \[\widehat{\psi}^{\pi_1}_{(1)} = \left[\begin{array}{cccc}
        \frac{1}{2} & 0 & -\frac{1}{2} & 0 \\
        0 & 0 & \frac{1}{\sqrt{2}} & 0
    \end{array}\right].\]
    Similarly, for $\pi_2 = (132)$, the corresponding $2\times 2\times 2$ hypermatrix is given by
    \[\widehat{\psi}^{\pi_2} = \left[\begin{array}{cc|cc}
        \frac{1}{2} & -\frac{1}{2} & 0 & 0 \\
        0 & \frac{1}{\sqrt{2}} & 0 & 0
    \end{array}\right]\]
    which in matrix form can be represented as
    \[\widehat{\psi}^{\pi_2}_{(1)} = \left[\begin{array}{cccc}
        \frac{1}{2} & -\frac{1}{2} & 0 & 0 \\
        0 & \frac{1}{\sqrt{2}} & 0 & 0
    \end{array}\right].\]
    The core tensor for each of these states is the same, which in matrix form is 
    \[\Sigma_{(1)} = \left[\begin{array}{cccc}
        \frac{\sqrt{2-\sqrt{2}}}{2} & 0 & 0 & 0 \\
        0 & \frac{\sqrt{2+\sqrt{2}}}{2} & 0 & 0
    \end{array}\right].\]
    Consequently, $\widehat{\psi}$, $\widehat{\psi}^{\pi_1}$, and $\widehat{\psi}^{\pi_2}$ are LU equivalent. Switching back to quantum states, this is the same as saying that the following $3$-qubits are LU equivalent
    \[\ket{\psi} = \frac{1}{2}\ket{000}-\frac{1}{2}\ket{100}+\frac{1}{\sqrt{2}}\ket{101}\]
    \[\ket{\psi}^{\pi_1} = \frac{1}{2}\ket{000}-\frac{1}{2}\ket{001}+\frac{1}{\sqrt{2}}\ket{101}\]
    \[\ket{\psi}^{\pi_2} = \frac{1}{2}\ket{000}-\frac{1}{2}\ket{010}+\frac{1}{\sqrt{2}}\ket{110}.\]
    
    On the other hand, consider the quantum state $\ket{\varphi} = \frac{1}{2}\ket{000}-\frac{1}{2}\ket{010}+\frac{1}{\sqrt{2}}\ket{101} = \left[\begin{array}{cccccccc}
        \frac{1}{2} & -\frac{1}{2} & 0 & 0 & 0 & \frac{1}{\sqrt{2}} & 0 & 0
    \end{array}\right]^t$. The corresponding $2\times 2\times 2$ hypermatrix $\widehat{\varphi}$ is given by
    \[\widehat{\varphi} = \left[\begin{array}{cc|cc}
        \frac{1}{2} & -\frac{1}{2} & 0 & 0 \\
        0 & 0 & \frac{1}{\sqrt{2}} & 0
    \end{array}\right].\]
    That is, $\widehat{\varphi}$ is obtained from $\widehat{\psi}$ by switching the $(211)$-coordinate of $\widehat{\psi}$ with its $(121)$-coordinate and vice versa, and leaving everything else fixed. Hence, there is no permuation relating $\widehat{\psi}$ with $\widehat{\varphi}$. Indeed, the matrix form of $\widehat{\varphi}$ can be represented as
    \[\widehat{\varphi}_{(1)} = \left[\begin{array}{cccc}
        \frac{1}{2} & -\frac{1}{2} & 0 & 0 \\
        0 & 0 & \frac{1}{\sqrt{2}} & 0
    \end{array}\right]\]
    and so it follows that the matrix form of its core tensor is
    \[\Sigma'_{(1)} = \left[\begin{array}{cccc}
        \frac{1}{\sqrt{2}} & 0 & 0 & 0 \\
        0 & \frac{1}{\sqrt{2}} & 0 & 0
    \end{array}\right].\]
    From this we see that $\widehat{\psi}$ and $\widehat{\varphi}$ have different $1$-mode singular values, proving that $\widehat{\psi}$ and $\widehat{\varphi}$ (and hence $\psi$ and $\varphi$) are not LU equivalent. 
\end{eg}
It just so happened in our example that the two states $\psi$ and $\varphi$, which were not related by a permutation, were not LU equivalent. We ask the following question: for any two quantum states that are not related by a permutation, are they necessarily not LU equivalent? If this is indeed true, then this would allow us to fully characterize entangled states in terms of the $\pi$-transpose, and additionally, it would make it very easy and quick to determine whether or not two states are LU equivalent. 

\section{Hyperdeterminants and $n$-tangles}
Recall that the {\it $n$-tangle}, a proposed measure of entanglement for pure $2n$-qubit states 
proposed in \cite{CW}, is defined as 
\[\tau_n(\ket{\psi}) = \left|\braket{\psi}{\widetilde{\psi}}\right|^2\]
where 
\[\ket{\widetilde{\psi}} = \sigma_y^{\otimes 2n},\]
with $\sigma_y$ being the second Pauli matrix 
$\left[\begin{array}{cc}
    0 & -i \\
    i & 0
\end{array}\right]$. The product $\sigma_y^{\otimes 2n}$ is sometimes known as the spin-flip transformation on $2n$-qubits. 
We now consider the relation between the $n$-tangle and the combinatorial hyperdeterminant via the hypermatrix of pure $2n$-qubit states. 

Let $\psi$ be a $2n$-qubit and $\widehat{\psi}$ be its corresponding hypermatrix as described in Section 
\ref{S3}. We now introduce an important matrix $\widehat{\mathrm{Ent}}_n$ as follows. Recall that the hyperdeterminant of $\widehat{\psi}$ is given by
\begin{equation}\mathrm{hdet}(\widehat{\psi}) = \sum\limits_{\sigma_2,...,\sigma_{2n}\in S_2}(-1)^m\psi_{0\sigma_2(0)...\sigma_{2n}(0)}\psi_{1\sigma_2(1)...\sigma_{2n}(1)}
\end{equation}
where $m$ denotes the number of permutations $\sigma_i\in S_2$ which are transpositions, and we will simply refer to the hyperdeterminant by $\mathrm{Ent}(\psi)$. Note in particular that since each $\sigma_i$ is in $S_2$, they are either the identity permutation (which takes $0$ to $0$ and $1$ to $1$) or they are the transposition which takes $0$ to $1$ and $1$ to $0$. Note also that $\mathrm{hdet}(\widehat{\psi})$ gives a quadratic form in the coefficients of $\psi$. 

Also, recall that for an arbitrary quadratic form 
\begin{equation}
    q(x_1,...,x_n) = \sum\limits_{i,j=1}^nq_{ij}x_ix_j,
\end{equation}
the \textit{matrix of the quadratic form} $q$ is the matrix $Q = [q_{ij}]\in \mathbb{C}^{n\times n}$. Denoting the vector $\left[\begin{array}{ccc}
    x_1 & ... & x_n
\end{array}\right]^t$ as $x$, we have that 
\begin{equation}
    x^tQx = q(x_1,...,x_n).
\end{equation}
We will denote the matrix of the quadratic form given by $\mathrm{Ent}(\psi)$ as $\widehat{\mathrm{Ent}}_n$. 
\begin{eg}
	For $n=1$, 
	\begin{align*}
	\mathrm{hdet}(\widehat{\psi}) &= \sum\limits_{\sigma_2\in S_2}(-1)^m\psi_{0\sigma_2(0)}\psi_{1\sigma_2(1)} \\
	&= \psi_{00}\psi_{11} - \psi_{01}\psi_{10}. 
	\end{align*}
	since $\sigma_2$ is either the identity or the only transposition in $S_2$. The matrix of this quadratic form is 
	\[\widehat{\mathrm{Ent}}_1 = \frac{1}{2}\left[\begin{array}{cccc}
	0 & 0 & 0 & 1 \\
	0 & 0 & -1 & 0 \\
	0 & -1 & 0 & 0 \\
	1 & 0 & 0 & 0
	\end{array}\right]\]
	For $n=2$, 
	\begin{align*}
	\mathrm{hdet}(\widehat{\psi}) &= \sum\limits_{\sigma_2,\sigma_3,\sigma_4\in S_2}(-1)^m\widehat{\psi}_{0\sigma_2(0)\sigma_3(0)\sigma_4(0)}\widehat{\psi}_{1\sigma_2(1)\sigma_3(1)\sigma_4(1)} \\
	&= \widehat{\psi}_{0000}\widehat{\psi}_{1111} - \widehat{\psi}_{0001}\widehat{\psi}_{1110} - \widehat{\psi}_{0010}\widehat{\psi}_{1101} + \widehat{\psi}_{0011}\widehat{\psi}_{1100} - \widehat{\psi}_{0100}\widehat{\psi}_{1011}\\
 &\qquad +\widehat{\psi}_{0101}\widehat{\psi}_{1010} + \widehat{\psi}_{0110}\widehat{\psi}_{1001} - \widehat{\psi}_{0111}\widehat{\psi}_{1000}. 
	\end{align*}
	The matrix of this quadratic form is given by
	\[\widehat{\mathrm{Ent}}_2 = \frac{1}{2}\left[\begin{array}{cccccccccccccccc}
	e_{16} & -e_{15} & -e_{14} & e_{13} & -e_{12} & e_{11} & e_{10} & -e_9 & -e_8 & e_7 & e_6 & -e_5 & e_4 & -e_3 & -e_2 & e_1 
	\end{array}\right]\]
	where $e_i$ is the $i^{th}$ vector in the standard ordered basis for $\mathbb{C}^{16}$. 
\end{eg}

From the above examples, we notice a few patterns. In general, each term in $\mathrm{hdet}(\widehat{\psi})$ is of the form $\pm \psi_{i...i_{2n}}\psi_{\overline{i_1}...\overline{i_{2n}}}$
where $\overline{i_j} = 1-i_j$. Equivalently, each term is of the form $\pm \ket{\psi}_j\ket{\psi}_{4^n-j+1}$ for $1\leq j\leq 4^n$. So after factoring out $\frac{1}{2}$ (which for the rest of this section we will assume we have already done), it follows that in general $\widehat{\mathrm{Ent}}_n$ is an anti-diagonal matrix with $1$'s and $-1$'s on its main anti-diagonal. 

Going from left to right, we represent each entry of the main anti-diagonal of $\widehat{\mathrm{Ent}}_n$ as a $+$ or $-$, with $1$ being identified as a $+$ and $-1$ being identified as a $-$. We then have that the main anti-diagonal of $\widehat{\mathrm{Ent}}_1$ is given by the string
\[+--+\]
In particular, the first entry gives the sign of the term $\psi_{00}\psi_{11}$, the $2^{nd}$ entry gives the sign of the term $\psi_{01}\psi_{10}$, the $3^{rd}$ entry gives the sign of the term $\psi_{01}\psi_{10}$, and the $4^{th}$ entry gives the sign of the term $\psi_{00}\psi_{11}$. Similarly, the main anti-diagonal of $\widehat{\mathrm{Ent}}_2$ is given by the string
\[+--+-++--++-+--+.\]
The first entry gives the sign of the term $\psi_{0000}\psi_{1111}$, the $2^{nd}$ entry gives the sign of the term $\psi_{0001}\psi_{1110}$,..., the $8^{th}$ entry gives the sign of the term $\psi_{0111}\psi_{1000}$, the $9^{th}$ entry gives the sign of the term $\psi_{0111}\psi_{1000}$,..., and the $16^{th}$ entry gives the sign of the term $\psi_{0000}\psi_{1111}$. By the hyperdeterminant formula, the sign of \[\ket{\psi}_j\ket{\psi}_{4^n-j+1} = \psi_{i_1...i_{2n}}\psi_{\overline{i_1}...\overline{i_{2n}}}\] 
is positive if there are an even number of $0$'s and $1$'s in either factor; likewise, the sign of 
\[\ket{\psi}_j\ket{\psi}_{4^n-j+1} = \psi_{i_1...i_{2n}}\psi_{\overline{i_1}...\overline{i_{2n}}}\] 
is negative if there is an odd number of $0$'s and $1$'s in either factor. Thus, $+$ corresponds to a coefficient of $\ket{\psi}$ with an even number of $0$'s and $1'$s, and $-$ corresponds to a coefficient of $\ket{\psi}$ with an odd number of $0$'s and $1$'s. 

Identify the coefficient $\psi_{i_1...i_{2n}}$ with the binary string $i_1...i_{2n}$, and let $B = b_1b_2...b_{4^n}$ denote the sequence consisting of all binary strings of length $2n$ ordered via the lexicographic order. We call a binary string $b_i$ "even" if it has an even number of $0$'s and $1$'s, and we call it "odd" if it has an odd number of $0$'s and $1$'s. Let $\chi$ be a function given by 
\begin{equation}\chi(b_i) = 
\begin{cases}
\begin{rcases}
+, & \text{ if }b_i\text{ is even} \\
-, & \text{ if }b_i\text{ is odd}
\end{rcases}.
\end{cases}
\end{equation}
Lastly, set 
\begin{equation}P := +--+
\end{equation}
and 
\begin{equation}N := -++-.
\end{equation}
\begin{fact}
	The binary string with a $1$ in only its $k^{th}$ position occurs in the $(2^{k-1}+1)^{th}$ position of $B$. 
\end{fact}
For $1\leq k\leq 2n$, call a binary string with only a $1$ in the $k^{th}$ position $k$. From Fact 1, in our notation, we have that 
\[k = b_{2^{k-1}+1}.\] 
So in particular, $3$ occurs after a sequence of $P$, $4$ occurs after a sequence of $PN$, $5$ occurs after a sequence of $PNNP$, $6$ occurs after a sequence of $PNNPNPPN$, and so on. Indeed, in general, we have the following result.
\begin{lemma}
	For $k\geq 3$, the binary string $k$ occurs after a sequence of $P$'s and $N$'s, which we denote as $S$. Moreover, the $k+1$ string occurs after the sequence $S\overline{S}$, where $\overline{S}$ is obtained after switching all $P$'s in $S$ to $N$, and likewise flipping all $N$'s in $S$ to $P$. 
\end{lemma}
\begin{proof}
First, note that $3 = b_5 = 0...0100$ occurs after a sequence of just $P$. This is because $b_1 = 0...0000$, $b_2 = 0...0001$, $b_3 = 0...0010$, $b_4 = 0...0011$, and so 
\begin{equation}\chi(b_1) = +,\qquad \chi(b_2) = -,\qquad \chi(b_3) = -,\qquad \chi(b_4) = +,
\end{equation}
which is precisely $P = +--+$. 
	
	Now, the string $b_{2^{k-1}+1+i}$ is obtained from the string $b_{1+i}$ after flipping the $k^{th}$ bit to a $1$, for $0\leq i\leq 2^{k-1}-1$. Therefore, if $\chi(b_{1+i}) = +$, then $\chi(b_{2^{k-1}+1+i}) = -$, and similarly if $\chi(b_{1+i}) = -$, then $\chi(b_{2^{k-1}+1+i}) = +$. Consequently, if it takes a sequence of $S$ (consisting of some ordering of $+$'s and $-$'s, which we assume nothing about) to get from $b_1$ up to but not including $k = b_{2^{k-1}+1}$, then it takes a sequence of $\overline{S}$ to get from $k = b_{2^{k-1}+1}$ up to but not including $k+1 = b_{2^k+1+i}$. That is, $k+1$ occurs after a string of $S\overline{S}$. From this and the fact that to get to $3$ it takes a sequence of $P$, it follows that $S$ is a sequence of $P$'s and $N$'s. 
\end{proof}
To recap, $\widehat{\mathrm{Ent}}_n$ is the matrix of the hyperdeterminant of the $2n$-qubit $\ket{\psi}$, whose coefficients $\psi_{i_1...i_{2n}}$ we have identified with the binary string $i_1...i_{2n}$, and each such string we have assigned a $+$ or $-$ to based on its parity. $\widehat{\mathrm{Ent}}_n$ is an anti-diagonal matrix whose main anti-diagonal can be represented as a sequence of $+$'s and $-$'s. 

Recall that the main anti-diagonal of $\widehat{\mathrm{Ent}}_1$ is given by $P$. After a sequence of $P$, we end up at the string $3 = 0...000100$. Therefore, by the lemma, after a sequence of $P\overline{P} = PN$, we end up at the string $4 = 0...001000$, and consequently after a sequence of $PN\overline{PN} = PNNP$, we end up at the string $5 = 0...010000$. Hence, the main anti-diagonal of $\widehat{\mathrm{Ent}}_2$ is given by 
\begin{equation}(P\overline{P})(\overline{P\overline{P}}) = PNNP
\end{equation}
Applying the same reasoning, it follows that the main anti-diagonal of $\widehat{\mathrm{Ent}}_3$ is given by 
\begin{equation}(PNNP\overline{PNNP})(\overline{PNNP\overline{PNNP}}) = PNNPNPPNNPPNPNNP.
\end{equation}
Indeed, continuing with this reasoning, in general, we have the following result. 
\begin{proposition}\label{p:prop2}
	The main anti-diagonal of 
 $\widehat{\mathrm{Ent}}_n$ is given by a sequence of $P$'s and $N$'s. Moreover, denoting its main anti-diagonal as $S$, we have that the main anti-diagonal of $\widehat{\mathrm{Ent}}_{n+1}$ is given by 
\begin{equation}(S\overline{S})(\overline{S\overline{S}}) = S\overline{S}\overline{S}S
\end{equation}
\end{proposition}
Since the second quarter of the main anti-diagonal of $\widehat{\mathrm{Ent}}$ is the negation of the first quarter, and since the second half of the main anti-diagonal of the negation of the first half, we have the following consequence. 
\begin{corollary}
	(After factoring out $\frac{1}{2}$) $\widehat{\mathrm{Ent}}_n$ is a symmetric anti-diagonal matrix whose main anti-diagonal consists of $1$'s and $-1$'s, and this holds for all positive integers $n$. 
\end{corollary}

Now we would like to study the relationship between 
the matrix of the hyperdeterminant of an arbitrary $2n$-qubit state with the spin-flip transformation.

To start with, we consider the structure of  $\sigma_y^{\otimes 2n}$. First note that 
\[\sigma_y^{\otimes 2} = \left[\begin{array}{cccc}
0 & 0 & 0 & -1 \\
0 & 0 & 1 & 0 \\
0 & 1 & 0 & 0 \\
-1 & 0 & 0 & 0
\end{array}\right],\]
which (like $\widehat{\mathrm{Ent}}$) is a symmetric anti-diagonal matrix consisting of $1$'s and $-1$'s. Going from left to right and representing each entry of the main anti-diagonal of $\sigma_y^{\otimes 2}$ as a $+$ or $-$, with $1$ being identified as $+$ and $-1$ being identified as $-$, we have that the main anti-diagonal of $\sigma_y^{\otimes 2}$ is given by 
\[-++-,\] 
which in our previous notation is just $N$. 
\begin{fact}
	Let $M$ be any arbitrary $n\times n$ anti-diagonal matrix with main anti-diagonal given by 
	\[(m_1,...,m_n) =: (m).\] Then $\sigma_2^{\otimes 2}\otimes M$ is given by the $4n\times 4n$ anti-diagonal matrix with main anti-diagonal given by 
\begin{equation}(-m_1,...,-m_n,m_1,...,m_n,m_1...,m_n,-m_1,...,-m_n) = (-m,m,m,-m).
\end{equation}
\end{fact}
From Fact 2 it follows that $\sigma_y^{\otimes 2n}$ is an anti-diagonal matrix. Furthermore, if the main anti-diagonal of $\sigma_y^{\otimes 2n}$ is denoted as $S$, then again by Fact 2 taking the Kronecker product of $\sigma_y^{\otimes 2}$ with $\sigma_y^{\otimes 2n}$ is equivalent to negating $S$, concatenating with $S$ twice, and then concatenating once more with the negation of $S$. Moreover, since the main anti-diagonal of $\sigma_y^{\otimes 2}$ is $P$, from this it follows that the main anti-diagonal of $\sigma_y^{\otimes 2n}$ is a sequence of $P$'s and $N$'s. In summary, we have the following proposition. 
\begin{proposition}\label{p:prop3}
	The main anti-diagonal of $\sigma_y^{\otimes 2n}$ is given by a sequence of $P$'s and $N$'s. Moreover, denoting its main anti-diagonal as $S$, we have that the main anti-diagonal of $\sigma_y^{\otimes 2(n+1)}$ is given by
\begin{equation}\overline{S}SS\overline{S}.
\end{equation}
\end{proposition}

We finally have everything we need to establish the equation relating the hyperdeterminant of $2n$-qubits with the Pauli matrix $\sigma_2$. 
\begin{theorem}
	Let $\widehat{\mathrm{Ent}}_n$ denote the matrix of the combinatorial hyperdeterminant of an arbitrary $2n$-qubit state $\ket{\psi}$. Then 
\begin{equation}\label{e:entequ}
\widehat{\mathrm{Ent}}_n = \frac{(-1)^n}{2}\sigma_y^{\otimes 2n}.
\end{equation}
\end{theorem}
\begin{proof}
	First, note that from Proposition \ref{p:prop2} and Proposition \ref{p:prop3}, we know that both $\widehat{\mathrm{Ent}}$ and $\sigma_y^{\otimes 2n}$ are anti-diagonal matrices whose main anti-diagonals are sequences of $P$'s and $N$'s. We proceed with induction. 
 For $n=1$, by direct computation, we have that the main anti-diagonal of $\widehat{\mathrm{Ent}}_1$ (after factoring out $\frac{1}{2}$) is $P$, and we also have that the main anti-diagonal of $\sigma_y^{\otimes 2}$ is $N$. Thus, 
 \begin{equation}
     \widehat{\mathrm{Ent}}_1 = -\frac{1}{2}\sigma_y^{\otimes 2}.
\end{equation}
Assume that the equation holds for some positive integer $n$. Now we consider the case of $n+1$. Denote the main anti-diagonal of $\widehat{\mathrm{Ent}}_n$ (after factoring out $\frac{1}{2}$) as $S$, and denote the main anti-diagonal of $\sigma_y^{\otimes 2n}$ as $T$. Then by the induction hypothesis, we have one of the following 2 cases:
\begin{enumerate}
	\item When $n$ is even, in which case by assumption we have that $S=T$. Then by Proposition \ref{p:prop2}, we have that the main anti-diagonal of $\widehat{\mathrm{Ent}}_{n+1}$ (after factoring out $\frac{1}{2}$) is given by 
	\[S\overline{S}\overline{S}S,\]
	and by Proposition \ref{p:prop3} we have that the main anti-diagonal of
    $\sigma_y^{2(n+1)}$ is given by
    \[
    \overline{T}TT\overline{T} = \overline{S}SS\overline{S} = \overline{S\overline{S}\overline{S}S}.
    \]
    Therefore,    
    \[
    \widehat{\mathrm{Ent}}_{n+1} = -\frac{1}{2}\sigma_y^{\otimes 2(n+1)}.
    \]
	\item When $n$ is odd, in which case by assumption $S = \overline{T}$. Then by Proposition \ref{p:prop2} we have that the main anti-diagonal of $\widehat{\mathrm{Ent}}_{n+1}$ (after factoring out $\frac{1}{2}$) is given by    
    \[
    S\overline{S}\overline{S}S,
    \]
	and by Proposition \ref{p:prop3} we have that the main anti-diagonal of $\sigma_y^{2(n+1)}$ is given by \[
    \overline{T}TT\overline{T} = S\overline{S}\overline{S}S.
    \]
Therefore, 
\[\widehat{\mathrm{Ent}}_{n+1} = \frac{1}{2}\sigma_y^{2(n+1)}.\]
\end{enumerate}
Combining the two cases we have that in general
\begin{equation}
\widehat{\mathrm{Ent}}_{n+1} = \frac{(-1)^{n+1}}{2}\sigma_y^{\otimes 2(n+1)}
\end{equation}
for any positive integer $n$. Thus the theorem is proved by induction.
\end{proof}

An almost immediate consequence is that the hyperdeterminant itself may be viewed as a measure of entanglement and an LU-invariant. 
\begin{corollary} We have that
\begin{equation}\tau_n(\ket{\psi}) = 4|\mathrm{hdet}(\widehat{\psi})|^2.
\end{equation}
\end{corollary}
\begin{proof}
    This is a straightforward calculation:
    \begin{align*}
        \tau_n(\ket{\psi}) &= \left|\braket{\psi}{\widetilde{\psi}}\right|^2 \\
        &= |\bra{\psi}{\sigma_y^{\otimes 2n}}\ket{\psi^*}|^2 \\
        &= 4|\bra{\psi}{\widehat{\mathrm{Ent}}_n}\ket{\psi^*}|^2,\quad \text{by Theorem 1} \\
        &= 4|\mathrm{hdet}(\widehat{\psi^*})|^2 \\
        &= 4|\mathrm{hdet}(\widehat{\psi})^*|^2,\quad \text{because in general }\mathrm{hdet}(H^*) = \mathrm{hdet}(H)^*\text{ for any cuboid hypermatrix }H \\
        &= 4|\mathrm{hdet}(\widehat{\psi})|^2.
    \end{align*}    
\end{proof}
A similar formula for the $n$-tangle involving determinants of the coefficients of $\psi$ was proven in \cite{Sharma-Sharma}, however by linking the $n$-tangle to the hyperdeterminant we can apply the theory of multilinear algebra to the $n$-tangle and more broadly the study of entanglement. For instance, it is known that the $n$-tangle is an LU-invariant, in fact, more generally a SLOCC invariant \cite{Li-Li}, and indeed this fact immediately follows from the above corollary since the hyperdeterminant is invariant under multilinear multiplication of matrices in the special linear group (Proposition \ref{p:prop1}).

\bigskip
\bigskip
\noindent\centerline{\bf Acknowledgments}

N. Jing is partially supported by Simons Foundation under the grant MP-TSM-00002518 during the work.

\bigskip
\bigskip
\noindent\centerline{\bf Data availability statement}

Any data that support the findings of this study are included within the article.

\bigskip
\bigskip


\end{document}